\documentclass[runningheads]{llncs}
\usepackage{graphicx}
\usepackage{amsfonts}
\begin{document}
\title{\vspace{-1cm}Worst-Case Analysis of LPT Scheduling on Small Number of Non-Identical Processors}

\titlerunning{Worst-Case Analysis of LPT Scheduling for Non-Identical Processors}
\author{
Takuto Mitsunobu\inst{1} \and
Reiji Suda\inst{1} \and Vorapong Suppakitpaisarn \inst{1}}
\authorrunning{T. Mitsunobu et al.}

\institute{\vspace{-0.3cm}The University of Tokyo, Japan}
\maketitle              % typeset the header of the contribution
\begin{abstract} \vspace{-0.5cm}
The approximation ratio of the longest processing time (LPT) scheduling algorithm has been studied in several papers. While the tight approximation ratio is known for the case when all processors are identical, the ratio is not yet known when the processors have different speeds. In this work, we give a tight approximation ratio for the case when the number of processors is 3,4, and 5. We show that the ratio for those cases are no more than the lower bound provided by Gonzalez, Ibarra, and Sahni (SIAM J. Computing 1977). They are approximately 1.38 for three processors, 1.43 for four processors, and 1.46 for five processors.\\
\textbf{Keywords:} Scheduling, LPT Algorithm, Approximation Ratio
\end{abstract}
\section{Introduction}
 \vspace{-0.2cm}
Theoretical analyses of algorithms for scheduling problems have been conducted in several recent works such as \cite{ghalami2019scheduling,jansen2020approximation} (refer to \cite{suda2006} for a survey). Several computation environments are considered in the works, but we will focus on the following environments in this paper:
\begin{itemize}
    \item \textbf{Offline:} All task information are available before the execution of a scheduling algorithm
    \item \textbf{Makespan minimization:} The scheduling algorithm aims to minimize the computation time needed by the processor that finish later than others.
    \item \textbf{No precedence constraints:} Tasks can be executed in any order. 
    \item \textbf{Uniform processor:} Let the size of task $i$ be $t(i)$ and let the speed of processor $p$ be $s(p)$.  Then, the calculation time of the task at the processor is $t(i)/s(p)$. 
\end{itemize}
The above setting is called as $Q || C_{\textrm{max}}$ in  scheduling literature \cite{suda2006}. The problem is NP-hard even when we have two processors with the same speed \cite{sahni1976algorithms,garey1978strong}. A polynomial time approximation scheme for this setting is known \cite{hochbaum1986polynomial}, but they are practically slow. We believe that it is more common to solve this problem using other scalable algorithms. Among those algorithms, the longest processing time (LPT) algorithm is one of the most well-known. The algorithm can be described as follows:
\begin{enumerate}
    \item Set $w(p) = 0$ for any processor $p$.
    \item Take $i$ as one of the largest tasks that have not yet been assigned to any processor.
    \item Take $p$ as one of processors which minimizes $w(p) + t(i)/s(p)$. 
    \item Assign task $i$ to processor $p$ and set $w(p)$ as $w(p) + t(i)/s(p)$.
    \item Go back to Step 2 until all tasks are assigned.
\end{enumerate}

\subsection{Previous Works}

The approximation ratio of the LPT algorithm is analyzed in several works. Dobson \cite{dobson1984scheduling} gave an instance where the approximation ratio of LPT is $1.512$ and proved the worst-case approximation ratio of LPT is not greater than $19/12 \approx 1.583$. Friesen \cite{friesen1987tighter} gave an instance where the approximation ratio of LPT is 1.52 and proved the worst-case approximation ratio of LPT is not greater than $5/3 \approx 1.667$ by different approach from \cite{dobson1984scheduling}. Later, Kov\'acs \cite{kovacs2010new} gave an instance where the worst-case approximation ratio of LPT is greater than $1.54$ and less than $1 + \sqrt{3}/3 \approx 1.577$.

The ratio for several special cases are also studied in several papers. Those include the case where only one processor has a
different speed \cite{gonzalez1977bounds,kovacs2009tighter}, the case all speeds of processors are a power of two \cite{kovacs2009tighter}, and the case where a ratio of speeds is a parameter \cite{chen1991parametric}.

In this paper, we consider a special case when the number of processors is a parameter. Suppose that the number of processors is $m$. We calculate the approximation algorithm in term of $m$. 
As the number of processors in distributed computation is usually small, we strongly believe that analyzing this special case is very important to understand the nature of the LPT algorithm.

Although its importance, there are only two previous works for this special case. Graham \cite{graham1966bounds,graham1969bounds} show that, when all processors are identical, the tight approximation ratio is $4/3 - 1/(3m)$. Gonzalez, Ibarra, and Sahni \cite{gonzalez1977bounds} show that the ratio is no larger than $2m/(m+1)$. Denote the unique positive root of the equation $2x^m - x^{m-1} - \cdots - x - 2 = 0$ by $\rho_m$. By a calculation, we have $\rho_2 \approx 1.28$, $\rho_3 \approx 1.38$, $\rho_4 \approx 1.43$, and $\rho_5 \approx 1.46$. The authors of \cite{gonzalez1977bounds} also gave a series of instances with $m$ processors where the approximation ratio of LPT is $\rho_m$. In addition, it is shown in the same paper that, when $m = 2$, the tight ratio is $\rho_2$.

\subsection{Our Contributions}

We show that the lower bound given by Gonzalez, Ibarra, and Sahni \cite{gonzalez1977bounds} is tight also for $m = 3,4,5$. In other words, $\rho_3$, $\rho_4$, and $\rho_5$ are tight approximation ratios when the numbers of processors are $3, 4,$ and $5$ respectively. 
Suppose that the number of tasks is $n$.
An informal sketch of our proof is as follows:
\begin{enumerate}
    \item In Section 3, we show that, if the worst approximation ratio is attained when $n = m + 1$, then the worst approximation is $\rho_m$.
    \item In Section 4, we show that the worst ratio is not attained when $n = m + 2$.
    \item In Section 5, we show that, when $m \in \{3,4,5\}$, the worst approximation ratio is attained only when $n \leq m + 2$.
\end{enumerate}

\section{Preliminaries}

%In this section, we introduce notations which we use throughout this paper. Then, we introduce the concept of minimality and domination which are important for our proof in the following section.

\subsection{Notations}

An instance for scheduling problems on uniform processors is defined as a 4-tuple
$I = (m, n, s, t)$ where $m$ is a number of processors, $n$ is a number of tasks, $s:\{1,\dots,m\} \rightarrow \mathbb{R}_+$ is
a function from indices of processors to their speeds, and $t:\{1,\dots,n\} \rightarrow \mathbb{R}_{\geq 0}$ is a function from indices of tasks to their sizes.

For an instance $I$, $OPT(I)$ denotes the optimal makespan and $LPT(I)$ denotes the makespan of the LPT schedule. $\rho_I = LPT(I)/OPT(I)$ is the approximation ratio. $A_I: \{1,\dots,n\} \rightarrow \{1, \dots,m\}$ is the assignment function by the LPT schedule. For all $p$, $T_I(p) =
\{i|A_I(i) = p\}$ is the set of tasks assigned to processor $p$ by the LPT schedule, and $w_I(p) = \sum_{i \in T_I(p)} t(i)$ is the workload of $p$ in the
 schedule. Similarly, $A^*_I: \{1,\dots,n\} \rightarrow \{1, \dots,m\}$ is the assignment function by an optimal schedule, $T^*_I(p) = \{i|A^*_I(i) = p\}$, and $w^*_I(p) = \sum_{i \in T^*_I(p)} t(i)$. In addition, $T'_I(p) = T^*_I(p) \backslash \{n\}$ and $w'_I(p) = \sum_{i \in T'_I(p)} t(i)$. It is quite important to consider the LPT result without the final task $n$ in our proof. The set $T'(p)$ and the workload $w_I'(p)$ will play an important role there.

Without loss of generality, we can assume that $OPT(I) = 1$, $s(1) \geq \dots \geq s(m)$ and $t(1) \geq \dots \geq t(n) = 1$ \cite{kovacs2010new}. 
We also know that there is an optimal assignment such that $w^*(1) \geq \dots \geq w^*(m)$ \cite{kovacs2010new}. 
We assume that $A^*$ is such a solution.  

Consider the equation $2x^m - x^{m-1} - \cdots - x - 2 = 0$. We know that there is only one positive solution by the Descartes’ rule of signs. Denote that solution by $\rho_m$. We know that $\rho_m$ is strictly increasing, i.e. $\rho_m > \rho_{m'}$ for $m > m'$. Gonzalez, Ibarra, and Sahni  has proved the following result in \cite{gonzalez1977bounds}.
\begin{lemma} \label{lem1}
There is an instance $I$ with $m$ processors and $m + 1$ tasks such that $\rho_I = \rho_m$. Moreover, none of  instance $I$ with $2$ processors has $\rho_I > \rho_2$.
\end{lemma}
We can follow an example given in \cite{friesen1987tighter} to show that when $m = 94$, the lower bound given by the above lemma is not tight. There is an instance $I$ with 94 processors such that $\rho_I > \rho_{94}$.  

\subsection{Minimality and Domination}
\begin{definition}[Minimality \cite{friesen1987tighter,kovacs2009tighter}]
An instance $I = (m, n, s, t)$ is minimal if 
$\rho_J < \rho_I$ for any instance 
$J = (m' , n' , s' , t')$ such that $m' \leq m$, $n' \leq n$, and $(m' , n')\neq (m, n)$.
\end{definition}
Intuitively, $I$ is minimal if $I$ is the smallest instance which can give a particular approximation ratio. During our proof, we will use the following lemma which is mentioned in \cite{kovacs2010new}.
\begin{lemma} \label{sec21}
If $I$ is minimal, then, for any processor $p$,
$(w_I'(p) + 1)/s(p) \geq \rho_I.$
\end{lemma}

Next, we define domination. The concept of domination was originally in \cite{coffman1978application}. 

\begin{definition}
In an instance $I$, a processor $p$ dominates a processor $q$ if $s(p) \leq s(q)$ and there is a function $f: T^*_I(q) \rightarrow T'_I(p)$ such that $t(i) \geq \sum_{f(j) = i} t(j)$ for any $i \in T'_I(p)$. Note that a processor may dominate itself.
\end{definition}

The following lemma, called the principle of domination, discusses the relationship between minimality and domination. 
%The principle was introduced in~\cite{friesen1987tighter},
%but we cannot verify the correctness of the proof in the paper. We, therefore, provide a proof for the LPT version here for a completeness of this paper.

\begin{lemma}[Principle of Domination]
If an instance I is minimal, no processor dominates any processor.
\end{lemma}

\section{Tight Approximation Ratio for $n = m + 1$}

We will work with the case when the number of tasks (denoted by $n$) is equal to the number of processors added by $1$ in this section. The following lemma, which is mentioned in  \cite{gonzalez1977bounds,kovacs2010new}, gives some properties of minimal instances.

\begin{lemma} \label{sec31}
If there is an empty processor in one of optimal schedules, the instance $I$ is not minimal. 
%If $n = m$ and each processor has just one task in one of optimal schedules for an instance $I$, $\rho_I = 1$.
Furthermore, if there exists a processor $p$ such that $T'_I(p) = \emptyset$, $I$ is not minimal
 or $\rho_I = 1$.
\end{lemma}

We prove the following lemma using the principle of domination.

\begin{lemma} \label{sec32}
If $I$ is minimal and $m \geq 2$, $|T^*_I(1)| \geq 2$.
\end{lemma}
\begin{proof}
The LPT algorithm always assigns the largest task (task 1) to the fastest processor (processor 1). Assume a contradictory statement that $|T^*_I(1)| = 1$ and the only task assigned to processor 1 is task $j$. Then, we know that processor 1 dominates itself, because we can have $f(j) = 1$, and $t(1) \geq \sum_{f(j')=1}t(j') = t(j)$. Thus, $I$ is not minimal. \qed
\end{proof}

We then prove the next theorem based on Lemmas \ref{sec31} and \ref{sec32}.

\begin{theorem}
If $I = (m, m + 1, s, t)$ is minimal, its approximation ratio of LPT on uniform processors is not greater than $\rho_m$. \label{thm1}
\end{theorem}
\begin{proof}
Let us first consider the LPT solution. As we know from Lemma \ref{sec31} that $T'_I(p) \neq \emptyset$ for all $p$, we have $T_I'(p) = \{p\}$ for all $1 \leq p \leq m$. We know from Lemma \ref{sec21} that $(t(p) + 1)/s(p) = (w_I'(p) + 1)/s(p) \geq \rho_I$ for all $p$.

Consider the optimal solution. From Lemma \ref{sec31}, there is no empty processor there, and we know that there are two tasks at the fastest processor by Lemma \ref{sec32}. Hence, $|T_I^*(p)| = 1$ for all $2 \leq p \leq m$. By the similar argument as in the proof of Lemma \ref{sec32}, we know that $T_I^*(p) = \{p - 1\}$ for all $2 \leq p \leq m$, otherwise the processor $p$ dominates itself. The only tasks left for the processor 1 are then $m$ and $m + 1$. By our assumption that $OPT(I) = 1$, we have $(t(m) + t(m + 1))/s(1) = (t(m) + 1)/s(1) \leq 1$ and $t(p - 1) / s(p) \leq 1$ for all $2 \leq p \leq m$. Hence, $s(1) \geq t(m) + 1$ and $s(p + 1) \geq t(p)$ for $1 \leq p < m$.

By merging the results at the last sentences of the first two paragraph, we obtain $s(1) / s(m) \geq \rho_I$ and $(s(p + 1) + 1)/s(p) \geq \rho_I$ for $p < m$. They are $s(p + 1) \geq \rho_I s(p) - 1$ and $s(1) \geq \rho_I s(m)$. Then, we can have $s(1) \geq \rho_I s(m) \geq \rho_I [\rho_I s(m - 1) - 1] \geq \cdots \geq \rho_I^m s(1) - \rho_I^{m - 1} - \dots - \rho_I.$
Hence, $\rho_I^m s(1) - \rho_I^{m - 1} - \dots - \rho_I - s(1) \leq 0$. Because $s(1) \geq t(m) + 1 \geq 2$, this implies that $\mathcal{P}(\rho_I) = 2\rho_I^m - \rho_I^{m - 1} - \dots - \rho_I - 2 \leq 0$. We obtain that $\rho_I \leq \rho_m$, because we know that all solutions of $\mathcal{P}$ would be no larger than $\rho_m$. \qed
\end{proof}

\section{Minimality of Instance $(m, m + 2, s,t)$}

We will show in this section that any instance with $m$ processors and $m + 2$ tasks is not minimal. We begin our proof by the following lemma.

\begin{lemma} \label{sec41}
Let I be a minimal instance. For any processor $p$ such that $s(p) > 1/(\rho_I - 1)$, we have $w'_I(p) > s(p)$. Moreover, for any processor $p$ such that
$s(p) \geq 1/(\rho_I - 1)$,
we have $w'_I(p) \geq s(p)$.
\end{lemma}
\begin{proof}
We know from Lemma \ref{sec21} that $w'_I(p) \geq \rho_I s(p) - 1$. Our assumption $s(p) > 1/(\rho_I - 1)$ can be written in the form of $\rho_I s(p) - 1 > s(p)$. We then obtain $w'_I(p) > s(p)$. We can use the same argument to show the latter part of the lemma statement. \qed
\end{proof}

We now suppose that there is a minimal instance with $n = m + 2$, and consider its LPT solution. Without the smallest task, we have $m + 1$ tasks here. By Lemma \ref{sec31}, we know that only one processor, denoted by processor $p_L$, is assigned to two tasks. The others are assigned to one. It is clear that the larger task assigned to $p_L$ is the $p_L$-th largest task. Suppose that the smaller task assigned to the processor is $\tau_L$. The task assigned to $p$ is then the $p$-th largest task for $p \in \{1, \dots, \tau_L - 1\}\backslash\{p_L\}$. For $p \geq \tau_L$, the task assigned to $p$ is the $(p + 1)$-th largest task. 

We consider the optimal solution of the instance from the next lemma. As all the tasks in $T'_I(p)$ is no smaller than the $(m + 1)$-th largest task, if the $(m + 1)$-th or the $(m + 2)$-th largest task is assigned as the only task of $T^*_I(p)$ for some $p$, then the processor $p$ dominates itself. We then know that the $(m + 1)$-th and the $(m + 2)$-th must be assigned to the processor with more than one tasks in the optimal solution. 

\begin{lemma} \label{sec42}
If $I = (m, m + 2, s, t)$ is minimal and $m \geq 3$, then $|T^*_I(1)| = 2$.
\end{lemma}
\begin{proof}
We know from Lemma \ref{sec32} that $|T^*_I(1)| \geq 2$. By Lemma \ref{sec31}, we obtain that $|T^*_I(1)| \in \{2,3\}$. It is then left to show that $|T^*_I(1)| \neq 3$. We assume a contradictory statement that $|T^*_I(1)| = 3$ and $|T^*_I(p)| = 1$ for $p > 1$. Then, $\{m + 1, m + 2\} \subseteq T^*_I(1)$. Denote the other task assigned to processor 1 by $\tau_O$.

As all task sizes are no less than 1, $w^*(1) = t(\tau_O) + t(m + 1) + t(m + 2) \geq 3$. Because $OPT(I) = 1$, $s(1) \geq w^*(1) \geq 3$. Because $I$ is minimal, we know from Lemma \ref{lem1} that $\rho_I \geq \rho_3 > 4/3$. Hence, $s(1) > 1/(\rho_I - 1)$, and, by Lemma \ref{sec41}, $w_I'(1) > s(1) \geq w_I^*(1)$. If there is only one task $i$ in $T_I'(1)$, then we have a self domination at the processor 1 by setting $f(j) = i$ for all $j \in T_I^*(1)$. Hence, there are two tasks in $T_I'(1)$ and $p_L = 1$. Furthermore, by $w_I'(1) > w_I^*(1)$, we have $t(1) + t(\tau_L) > t(\tau_O) + t(m + 1) + t(m + 2)$. If $\tau_L = m + 1$, we have a self-domination at processor $1$ by having $f(\tau_O) = f(m + 2) = 1, f(m + 1) = m + 1$. Therefore, $\tau_L \leq m$. 

We know that $T_I'(\tau_L) = \{\tau_L + 1\}$. To avoid the self-domination at the $\tau_L$-th fastest processor, the only task in $T_I^*(\tau_L)$ must be larger than the $(\tau_L + 1)$-largest task. To keep makespan of the optimal solution equal 1, we have $s(\tau_L) \geq t(\tau_L)$.
When the LPT algorithm chooses a processor to process the $\tau_L$-th largest task, it prefers processor 1 over the $\tau_L$-th fastest processor. We, therefore, have $(t(1) + t(\tau_L))/s(1) \leq t(\tau_L)/s(\tau_L) \leq 1$ and $w'_I(1) = t(1) + t(\tau_L) \leq s(1)$. That contradicts $w_I'(1) > s(1)$ which we show earlier. \qed
\end{proof}

We can conclude from Lemma \ref{sec42} that, in addition to processor 1, there is another processor with two tasks in the optimal solution. Suppose that it is the $p_O$-th fastest processor. We show a property of $p_O$ in the next lemmas.

\begin{lemma} \label{sec43}
If $I = (m, m + 2, s, t)$ is minimal and $m \geq 3$, then $p_O > p_L$.
\end{lemma}
\begin{proof}
Assume a contradictory statement that $p_O \leq p_L$.
Recall that $T'_I(p_L) = \{p_L, \tau_L\}$. Also, recall that the task $m + 1$ and $m + 2$ must be assigned to either processor $1$ or processor $p_O$. If $T'_I(1)$ or $T'_I(p_O)$ is $\{m + 1, m + 2\}$, then, by having $f(m + 1) = f(p_L), f(m + 2) = f(\tau_L)$, we have the processor $p_L$ dominate processor 1 or processor $p_O$. To make sure that there is not such a domination, both processors $1$ and $p_O$ must be given a task with size larger than the $p_L$-th largest task in the optimal solution. Similarly, for all processor $p \in \{1, \dots, p_L\} \backslash \{1, p_O\}$, the only task assign to $p$ in the optimal solution must be larger than the $p_L$-th largest task, otherwise $p$ is dominated by the processor $p_L$. Each of the fastest $p_L$ processors need a task larger than the $p_L$-th largest task in the optimal solution. This contradicts the fact that there are only $p_L - 1$ of them. \qed
\end{proof}

\begin{lemma}
If $I = (m, m + 2, s, t)$ is minimal and $m \geq 3$, then $p_O < \tau_L$. 
\end{lemma}
\begin{proof}
Assume a contradictory statement that $p_O \geq \tau_L$. Recall that $T'_I(p) = \{p\}$ for all $p \in \{1, \dots, \tau_L - 1\} \backslash \{p_L\}$ and $T'_I(p_L) = \{p_L, \tau_L\}$. By the assumption, we know that $|T^*_I(p)| = 1$ for all $2 \leq p \leq \tau_L - 1$. To avoid the self-domination, we must have $T^*_I(p) = \{p - 1\}$ for $2 \leq p \leq \tau_L - 1$. We then know that the larger task of $T^*_I(1)$, denoted by $j_\ell$, is not one of the the $(\tau_L - 2)$-th largest task, and is no larger than the $p_L$-th largest task. On the other hands, the smaller task of $T^*_I(1)$, denoted by $j_s$ is not one of the $(\tau_L - 1)$-th largest, and is no larger than the $\tau_L$-th largest task. The processor 1 is then dominated by the $p_L$-fastest processor as we can have $f(j_\ell) = p_L$ and $f(j_s) = \tau_L$. \qed
\end{proof}

\begin{lemma}
If $I = (m, m + 2, s, t)$ is minimal and $m \geq 3$, then $T^*_I(p_O) = \{m + 1, m + 2\}$. 
\end{lemma}
\begin{proof}
Recall that $T'_I(p) = \{p\}$ for all $p \in \{1, \dots, \tau_L - 1\} \backslash \{p_L\}$ and $T'_I(p_L) = \{p_L, \tau_L\}$. Also, $T^*_I(p) = \{p - 1\}$ for all $2 \leq p < p_O$ to avoid the self-domination. The larger task assigned to processor $1$ in the optimal solution, denoted by $j_\ell$ is then not among the $p_O - 1$ largest task. By Lemma \ref{sec43}, it is smaller than the task $p_L$. If the smaller task, denoted by $j_s$ is the $(m + 1)$-th or $(m + 2)$-th, then processor $p_L$ dominates processor $1$ by setting $f(j_\ell) = p_L$ and $f(j_s) = \tau_L$. However, we discuss earlier that both the $(m + 1)$-th and $(m + 2)$-th must be assigned to processor 1 or processor $p_O$. As they are not assigned to processor 1, they must be in $T_I^*(p_O)$. \qed
\end{proof}

From the next lemma, we will consider the properties of tasks assigned to processor 1 in the optimal solution. Suppose that they are the $\tau_\ell$-th and $\tau_s$-th largest task, when $\tau_\ell < \tau_s$.
\begin{lemma} \label{lem11}
If $I = (m, m + 2, s, t)$ is minimal and $m \geq 3$, then $p_O - 1 \leq \tau_\ell \leq \tau_L - 2$ and $\tau_s = \tau_L - 1$. 
\end{lemma}
\begin{proof}
Recall that $T^*_I(p) = \{p - 1\}$ for all $2 \leq p < p_O$ to avoid the self-domination. Therefore, by Lemma \ref{sec43}, $\tau_\ell \geq p_O - 1 \geq p_L$.

Also, $T_I'(p_L) = \{p_L, \tau_L\}$ and $T_I'(\tau_L - 1) = \{\tau_L - 1\}$. To avoid the domination from processor $\tau_L - 1$, the only element of processor $p \in \{2, \dots, \tau_L - 1\}\backslash\{p_O\}$ must be larger than the $(\tau_L - 1)$-th largest task. All but one of the $\tau_L - 2$ largest tasks are used there. We know that $\tau_\ell$ is the task that is not used and $\tau_s = \tau_L - 1$, otherwise we would have $\tau_\ell \leq p_L$ and $\tau_s \leq \tau_L$. In that case,  processor~1 is dominated by processor $p_L$. \qed
\end{proof}

We will now prove the following lemmas to prepare for Lemma \ref{lem13}.

\begin{lemma} \label{lem12}
If $I = (m, m + 2, s, t)$ is minimal, $m \geq 3$, and $\tau_L \neq m + 1$, then $w_I'(p_L) \leq s(p_L)$. 
\end{lemma}
\begin{proof}
If $\tau_L \leq m$, $T_I'(p_L) = \{p_L, \tau_L\}$ and $T_I'(\tau_L) = \{\tau_L + 1\}$. On the other hand, $T_I^*(p_L) = \{p_L - 1\}$ and $T^*_I(\tau_L) = \{\tau_L\}$. Because the LPT algorithm chooses to process the task $\tau_L$ at the processor $p_L$ not $\tau_L$, we know that $w_I'(p_L) / s(p_L) = (t(p_L) + t(\tau_L))/s(p_L) \leq t(\tau_L)/s(\tau_L) = w^*_I(\tau_L)/s(\tau_L) \leq 1$. From there, we can conclude that $w_I'(p_L) \leq s(p_L)$. \qed
\end{proof}

\begin{lemma} \label{lem122}
If $I = (m, m + 2, s, t)$ is minimal and $m \geq 3$, $\rho_I < 1.5$.
\end{lemma}
\begin{proof}
Assume that $I$ is a minimal instance such that $\rho_I \geq 1.5$, then, for any $p \leq p_O$, $1/(\rho_I - 1) \leq 2 \leq w_I^*(p_O)\leq w_I^*(p) \leq s(p)$. Thus, by Lemma \ref{sec41}, $w'_I(p) \geq s(p) \geq w_I^*(p)$. Since $p_O \geq 2$ and there is only one processor $p$ for which $|T_I'(p)| \geq 2$, there exists a processor $p \leq p_O$ for which $|T_I'(p)| = 1$. Let the only task of that processor be $i$. If $f(j) = i$ for all $j \in T^*_I(p)$, we have $w'_I(p) = t(i) \geq \sum_{f(j) = i}f(j) = w_I^*(p)$. This means that $p$ dominates itself and $I$ is not minimal. \qed
\end{proof}

In the following lemma, we consider a polynomial $\mathcal{P}_m(x) = 2x^m - x^{m -1} - \dots - x - 2$. It is known that, for any set $S \subseteq \{1, \dots, m\}$, all positive solutions of $\sum_{i\in S} \mathcal{P}_i(x) \leq 0$ are no larger than $\rho_m$. On the other hand, all positive solutions of $\sum_{i\in S} \mathcal{P}_i(x) + 3 - 2x \leq 0$ are in $(0,\rho_m] \cup [1.5, \infty)$.
\begin{lemma} \label{lem13}
If $I = (m, m + 2, s, t)$ is minimal and $m \geq 3$, $\tau_\ell \geq p_O + 1$.
\end{lemma}
\begin{proof}
Assume contradictory statement that $\tau_\ell \leq p_O$. Then, $w_I^*(1) = t(\tau_\ell) + t(\tau_s) \geq  t(p_O) + 1$. On the other hand, as $s(p_O) \geq |T^*_I(p_O)| = 2$, by Lemma \ref{sec21}, we have $\rho_I \leq (w_I'(p_O) + 1)/s(p_O) = (t(p_O) + 1)/s(p_O) \leq w_I^*(1) / 2$. Because of the minimality of $I$ and Lemma \ref{lem1}, $\rho_I \geq \rho_3 > (1 + \sqrt{3}) / 2$ and $1/(\rho_I - 1) < 2\rho_I \leq w_I^*(1) \leq s(1)$. By Lemma \ref{sec41}, we know that $w_I'(1) > s(1)$. If there is only one task in $T_I'(1)$, then $w_I'(1) > s(1) \geq w_I^*(1)$, processor 1 dominates itself. Therefore, $p_L = 1$ and $w_I'(p_L) > s(p_L)$. By Lemma \ref{lem12}, we obtain that $\tau_L = m + 1$, $T_I'(1) = \{1, m + 1\}$ and $T_I'(p) = \{p\}$ for $p \geq 2$. By Lemma \ref{sec21},
\begin{equation} \label{eqn1}
    t(1) + t(m + 1) + 1 \geq s(1)\rho_I \textrm{~~~and~~~~} t(p) + 1 \geq s(p) \rho_I \textrm{~~~for all } p \geq 2.
\end{equation}
From Lemma \ref{lem11}, it is left to show that $\tau_\ell \in \{p_O - 1,p_O\}$ leads to contradictions. We then consider the following four cases:
\begin{case}[$\tau_\ell = p_O - 1, p_O = m$] We then have $T_I^*(1) = \{m - 1, m\}$, $T_I^*(p) = \{p - 1\}$ for $2 \leq p \leq m - 1$, and $T_I^*(m) = \{m + 1, m + 2\}$. From (\ref{eqn1}) and the fact that $s(p) \geq w_I^*(p)$ for all $p$, we obtain:
$t(m - 1)  \geq  s(m - 1)\rho_I - 1 \geq t(m - 2)\rho_I - 1\geq (s(m - 2)\rho_I - 1)\rho_I - 1 
 \geq  \cdots \geq s(2)\rho_I^{m - 2} - \rho_I^{m - 3} - \cdots - 1.$
 
We then have $s(2) + s(m) \geq t(1) + t(m + 1) + 1 \geq s(1)\rho_I \geq (t(m-1) + t(m))\rho_I \geq s(2)\rho_I^{m - 1} - \rho_I^{m - 2} - \cdots - \rho_I + s(m)\rho_I^2 - \rho_I$. Since $s(2) \geq s(m) \geq 2$, we have $\mathcal{P}_{m - 1}(\rho_I) + \mathcal{P}_2(\rho_I) \leq 0$. Then, $\rho_I \leq \rho_m$, which contradicts the fact that $I$ is minimal.
\end{case}
\begin{case}[$\tau_\ell = p_O - 1, p_O < m$] We then have $T_I^*(1) = \{p_O - 1, m\}$, $T_I^*(p) = \{p - 1\}$ for $p \in \{2, \dots, m\} \backslash \{p_O\}$, and $T_I^*(p_O) = \{m + 1, m + 2\}$. By (\ref{eqn1}),
\begin{eqnarray*}
t(p_O - 1) & \geq & s(p_O - 1)\rho_I - 1 \geq t(p_O - 2)\rho_I - 1\geq (s(p_O - 2)\rho_I - 1)\rho_I - 1 \\
& \geq & \cdots \geq s(2)\rho_I^{p_O - 2} - \rho_I^{p_O - 3} - \cdots - 1.\\
t(m) & \geq & s(p_O+1)\rho_I^{m - p_O} - \rho_I^{m - p_O - 1} - \cdots - 1.
\end{eqnarray*}
By $s(p_O + 1) \geq t(p_O) \geq s(p_O)\rho_I - 1$ and Lemma \ref{sec21}, we have 
\begin{eqnarray*}
& & s(2) + s(p_O) \geq t(1) + t(m + 1) + 1 \geq s(1)\rho_I \geq (t(p_O - 1) + t(m))\rho_I \\
& \geq & s(2)\rho_I^{p_O - 1} - \rho_I^{p_O - 2} - \cdots - \rho_I + s(p_O+1)\rho_I^{m - p_O + 1} - \rho_I^{m - p_O} - \cdots - \rho_I \\
& \geq & s(2)\rho_I^{p_O - 1} - \rho_I^{p_O - 2} - \cdots - \rho_I + s(p_O)\rho_I^{m - p_O + 2} - \rho_I^{m - p_O + 1} - \cdots - \rho_I.
\end{eqnarray*}
Because $s(2) \geq s(p_O) \geq 2$, we can conclude that $\mathcal{P}_{p_O - 1}(\rho_I) + \mathcal{P}_{m - p_O + 2}(\rho_I) \leq 0$ and $\rho_I \leq \rho_m$, which contradicts the fact that $\rho_I > \rho_m$.
\end{case}
\begin{case}[$\tau_\ell = p_O = 2$] We then have $T_I^*(1) = \{2,m\}$, $T_I^*(2) = \{m + 1,m + 2\}$, $T_I^*(3) = \{1\}$, and $T_I^*(p) = \{p - 1\}$ for $p \geq 4$. By (\ref{eqn1}), $t(m) \geq s(3)\rho_I^{m - 2} - \rho_I^{m - 1} - \cdots - 1$. Thus, $s(2) + s(3) \geq t(1) + t(m + 1) + 1 \geq s(1)\rho_I \geq (t(2) + t(m))\rho_I \geq s(2)\rho_I^2 - \rho_I + s(3)\rho_I^{m - 1} - \rho_I^{m - 2} - \cdots - \rho_I$.
Because $s(2) \geq 2$, $s(3) \geq t(1) \geq t(2) \geq s(2)\rho_I - 1$, and $\rho_I > 1$,
\begin{eqnarray*}
0 & \geq & s(2)(\rho_I^2 - 1) - \rho_I + s(3)(\rho_I^{m - 1} - 1) - \rho_I^{m - 2} - \cdots - \rho_I \\
& \geq & s(2)(\rho_I^2 - 1) - \rho_I + (s(2)\rho_I - 1)(\rho_I^{m - 1} - 1) - \rho_I^{m - 2} - \cdots - \rho_I \\
& \geq & 2\rho_I^2 - \rho_I - 2 + 2\rho_I^m - \rho_I^{m - 1} - \cdots - \rho_I - 2 + 3 - 2\rho_I \\ & = & \mathcal{P}_2(\rho_I) + \mathcal{P}_m(\rho_I) + 3 - 2\rho_I.
\end{eqnarray*}
We then have $\rho_I \in (0, \rho_m] \cup [1.5, \infty)$, which contradicts Lemma \ref{sec21} or Lemma~\ref{lem122}.
\end{case}
\begin{case}[$\tau_\ell = p_O$, $p_O \neq 2$] We then have $T_I^*(1) = \{p_O,m\}$, $T_I^*(p) = \{p - 1\}$ for $2 \leq p < p_O$, $T_I^*(p_O) = \{m + 1,m + 2\}$, $T_I^*(p_O + 1) = \{p_O - 1\}$, and $T_I^*(p) = \{p - 1\}$ for $p \geq p_O + 2$.
Also, $p_O = \tau_\ell < \tau_s = m$. Then, by (\ref{eqn1}), 
\begin{eqnarray*}t(m) & \geq & s(m)\rho_I - 1 \geq t(m - 1) \rho_I - 1 \geq (s(m - 1)\rho_I) - 1) \rho_I - 1 \geq \cdots\\ 
& \geq & s(p_O + 1)\rho_I^{m - p_O} - \rho_I^{m - p_O - 1} - \cdots - 1 \\
& \geq & t(p_O - 1)\rho_I^{m - p_O} - \rho_I^{m - p_O - 1} - \cdots - 1 \\ & \geq & \cdots  \geq  s(2) \rho_I^{m - 2} - \rho_I^{m - 3} - \cdots - 1.
\end{eqnarray*}
Thus, $s(2) + s(p_O) \geq t(1) + t(m + 1) + 1 \geq s(1)\rho_I \geq (t(p_O) + t(m))\rho_I \geq s(p_O)\rho_I^2 - \rho_I + s(2) \rho_I^{m - 1} - \rho_I^{m - 2} - \cdots - \rho_I$. Since $s(2) \geq s(p_O) \geq 2$, we obtain $\mathcal{P}_2(\rho_I) + \mathcal{P}_{m - 1}(\rho_I)$. That means $\rho_I \leq \rho_m$, which contradicts Lemma \ref{lem1}.  \qed
\end{case}
\end{proof}
By previous lemmas, we have obtain several properties of minimal instances when $n = m + 2$. We are now ready to show in the next theorem that there is no instance with such properties, and, hence, an
instance $(m, m + 2, s, t)$ is not minimal.
\begin{theorem}
Any instance $I = (m, m + 2, s, t)$ is not minimal. \label{thm2}
\end{theorem}
\begin{proof}
Assume that $I$ is minimal. We will consider the following four cases:
\begin{case}[$p_L = p_O - 1, \tau_L = m + 1$] Recall Lemma \ref{lem13} that $\tau_\ell > p_O$. We then have $T'_I(p) = \{p\}$ for $p \neq p_O - 1$, $T'_I(p_O - 1) = \{p_O - 1, m + 1\}$, $T_I^*(1) = \{\tau_\ell, m\}$, $T_I^*(p) = \{p - 1\}$ for $p \in \{2, \dots, m\} \backslash \{p_O\}$, and $T_I^*(p_O) = \{m + 1, m + 2\}$. Then,
\begin{eqnarray}
& & s(p_O) + s(p_O + 1)  \geq t(p_O - 1) + t(m + 1) + 1 \geq s(p_O - 1) \rho_I \nonumber \\ & \geq & s(1) \rho_I^{p_O - 1} - \rho_I^{p_O - 2} - \cdots - \rho_I \geq (t(\tau_\ell) + t(m)) \rho_I^{p_O - 1} - \rho_I^{p_O - 2} - \cdots - \rho_I \nonumber \\
& \geq & s(\tau_\ell)\rho_I^{p_O} - \rho_I^{p_O - 1} + s(\tau_\ell + 1)\rho_I^{m - \tau_\ell + p_O - 1} - \rho_I^{m - \tau_\ell + p_O - 2} - \dots - \rho_I^{p_O - 1} \nonumber \\
& & - \rho_I^{p_O - 2} - \dots - \rho_I \nonumber \\
& \geq & s(\tau_\ell)\rho_I^{p_O} - \rho_I^{p_O - 1} - \dots - \rho_I \nonumber \\ & & + s(\tau_\ell + 1)\rho_I^{m - \tau_\ell + p_O - 1} - \rho_I^{m - \tau_\ell + p_O - 2} - \dots - \rho_I \label{eqn2}
\end{eqnarray}
Now, let us consider processor $p$ such that $p_O < p \leq \tau_\ell + 1$. We know that there is a processor $p'$ such that $p' \geq p$ and the only task in $T^*_I(p')$ is one of the $p - 2$ largest. Therefore,  $s(p) \geq s(p') \geq t(p - 2) \geq s(p - 2)\rho_I - 1$. For some $\alpha, \beta$, one of $s(\tau_\ell)$ and $s(\tau_\ell + 1)$ can be written in the form of $s(p_O - 1)\rho_I^\alpha - \rho_I^{\alpha - 1} \cdots - 1$, while the other can be written in the $s(p_O)\rho_I^\beta - \rho_I^{\beta - 1} \cdots - 1$. 
Applying these inequalities to (\ref{eqn2}), we know that there are $\alpha',\beta' \leq m - 1$ 
such that 
$s(p_O) + s(p_O + 1)  
\geq s(p_O) \rho_I^{\alpha'} - \rho_I^{\alpha' - 1} - \cdots - \rho_I + s(p_O + 1) \rho_I^{\beta'} - \rho_I^{\beta' - 1} - \cdots - \rho_I$. Since $s(p_O + 1) \geq s(p_O - 1) \rho_I - 1 \geq s(p_O)\rho_I - 1$, we obtain: \label{case5}
\begin{eqnarray*}
0 & \geq & s(p_O)(\rho_I^{\alpha'} - 1) - \rho_I^{\alpha' - 1} - \cdots - \rho_I + s(p_O + 1)(\rho_I^{\beta'} - 1) - \rho_I^{\beta' - 1} - \cdots - \rho_I \\
& \geq & 2\rho_I^{\alpha'} - \rho_I^{\alpha' - 1} - \cdots - \rho_I - 2 + (s(p_O)\rho_I - 1)(\rho_I^{\beta'} - 1) - \rho_I^{\beta' - 2} - \cdots - \rho_I \\
& \geq & 2\rho_I^{\alpha'} - \rho_I^{\alpha' - 1} - \cdots - \rho_I - 2 + 2\rho_I^{\beta' + 1} - \rho_I^{\beta'} - \cdots - \rho_I - 2 + 3 - 2\rho_I \\ 
& = & \mathcal{P}_{\alpha'}(\rho_I) + \mathcal{P}_{\beta' + 1}(\rho_I) + 3 - 2\rho_I.
\end{eqnarray*}
We then have $\rho_I \in (0, \rho_m] \cup [1.5, \infty)$, which contradicts Lemma \ref{sec21} or Lemma~\ref{lem122}.
\end{case}
\begin{case}[$p_L \neq p_O - 1, \tau_L = m + 1$] We then have $T'_I(p) = \{p\}$ for $p \neq p_L$, $T'_I(p_L) = \{p_L, m + 1\}$, $T_I^*(1) = \{\tau_\ell, m\}$, $T_I^*(p) = \{p - 1\}$ for $p \in \{2, \dots, m\} \backslash \{p_O\}$, and $T_I^*(p_O) = \{m + 1, m + 2\}$. Note that $p_L < p_O - 1$  here. By that,
\begin{eqnarray}
& & s(p_L + 1) + s(p_O)  \geq t(p_L) + t(m + 1) + 1 \geq s(p_L) \rho_I \nonumber \\ & \geq & s(1) \rho_I^{p_L} - \rho_I^{p_L - 1} - \cdots - \rho_I \geq (t(\tau_\ell) + t(m)) \rho_I^{p_L} - \rho_I^{p_L - 1} - \cdots - \rho_I \nonumber \\
& \geq & s(\tau_\ell)\rho_I^{p_L + 1} - \rho_I^{p_L} + s(\tau_\ell + 1)\rho_I^{m - \tau_\ell + p_L} - \rho_I^{m - \tau_\ell + p_L - 1} - \dots - \rho_I^{p_L} \nonumber \\
& & - \rho_I^{p_L - 1} - \dots - \rho_I \nonumber \\
& \geq & s(\tau_\ell)\rho_I^{p_L + 1} - \rho_I^{p_L} - \dots - \rho_I \nonumber \\ & & + s(\tau_\ell + 1)\rho_I^{m - \tau_\ell + p_L} - \rho_I^{m - \tau_\ell + p_O - 1} - \dots - \rho_I \label{eqn3}
\end{eqnarray}
Similar to Case \ref{case5}, we have $s(p) \geq s(p - 2) \rho_I - 1$ for $p_O + 1 \leq p \leq \tau_\ell + 1$. Furthermore, $s(p_O - 1) \geq s(p_L + 1)\rho_I^{p_O - p_L - 2} - \rho_I^{p_O - p_L - 3} - \cdots - 1$. By applying those facts to Equation (\ref{eqn3}), we obtain $s(p_L + 1) + s(p_O) \geq s(p_L + 1)\rho_I^\alpha - \rho_I^{\alpha - 1} - \cdots - \rho_I + s(p_O)\rho_I^\beta - \rho_I^{\beta - 1} - \rho_I$ for some $\alpha, \beta < m$. Because $s(p_L + 1) \geq s(p_O) \geq 2$, $2\rho_I^\alpha - \rho_I^{\alpha - 1} - \cdots - \rho_I - 2 + 2\rho_I^{\beta} - \rho_I^{\beta - 1} - \cdots - \rho_I - 2 \leq 0$ and $\mathcal{P}_\alpha(\rho_I) + \mathcal{P}_\beta(\rho_I) \leq 0$. This means $\rho_I \leq \rho_m$, which contradicts the fact that $I$ is minimal.
\end{case}
\begin{case}[$p_L = p_O - 1$ and $\tau_L \leq m$] We then have $T'_I(p) = \{p\}$ for $p \neq p_O - 1$, $T'_I(p_L) = \{p_L, \tau_L \}$, $T_I^*(1) = \{\tau_\ell, \tau_L - 1\}$, $T_I^*(p) = \{p - 1\}$ for $p \in \{2, \dots, m\} \backslash \{p_O\}$, and $T_I^*(p_O) = \{m + 1, m + 2\}$. Recall from Lemma \ref{lem11} that $\tau_\ell \leq \tau_L - 2$. We obtain $t(\tau_L - 1) \geq s(\tau_\ell + 1) \rho_I^{\tau_L - \tau_\ell - 1} - \rho_I^{\tau_L - \tau_\ell - 2} - \cdots - 1$. Similar to Case \ref{case5}, for $p_O + 2 \leq p \leq \tau_\ell + 1$, $s(p) \geq s(p - 2) \rho_I - 1$. Furthermore, $s(p_O + 1) \geq t(p_O - 1) \geq t(p_O) \geq s(p_O) \rho_I - 1$. Thus, for some $\alpha < m$, $t(\tau_L - 1) \geq s(p_O)\rho_I^\alpha - \rho_I^{\rho - 1} - \cdots - 1 \geq 2\rho_I^\alpha - \rho_I^{\rho - 1} - \cdots - 1$. If $p_O = 2$, then $p_L = 1$ and $w_I'(1) \leq s(1)$ by Lemma \ref{lem12}. Otherwise, we have $w_I'(1) < w_I^*(1) \leq s(1)$ unless processor $1$ dominates itself. Then, by Lemma \ref{sec41}, $s(1) \leq 1/(\rho_I - 1)$ and
$\frac{1}{2(\rho_I - 1)} \geq \frac{s(1)}{2} \geq \frac{t(\tau_\ell) + t(\tau_L - 1)}{2} \geq t(\tau_L - 1) \geq 2\rho_I^\alpha - \rho_I^{\alpha - 1} - \cdots - 1.$
Thus, $2\rho_I^{\alpha + 1} - \rho_I^{\alpha} - \cdots - \rho_I - 2 + (3\rho_I - 4)/(2\rho_I - 2) \leq 0$. All solutions of the inequality is no more than $\max\{4/3, \rho_m\}$. However, this contradicts with the fact that $\rho_I > \rho_m \geq 4/3$.
\end{case}
\begin{case} [$p_L \neq p_O - 1, \tau_L \leq m$] By Lemma \ref{lem12}, we know that $w_I'(p_L) \leq s(p_L)$. Because $T_I^*(p_L + 1) = \{p_L\}$, we have $t(p_L) = w_I^*(p_L + 1) \geq w_I^*(p_O)$. Hence, $s(p_L) \geq t(p_L) + t(\tau_L) \geq t(m + 1) + t(m + 2) + t(\tau_L) \geq 3$ and $s(p_L) > 1/(\rho_I - 1)$. By Lemma \ref{sec41}, $w'_I(p_L) > s(p_L)$ which contradicts with the fact obtained from Lemma \ref{lem12}.  \qed

\end{case}
\end{proof}
\section{Tight Approximation Ratios for $3, 4, 5$ Processors}

We show the main result of this paper in this section. The proof begins with the following lemma.
\begin{lemma}\label{lem15}
If $I = (m,n,s,t)$ is minimal, then $n \leq \sum_i t(i) \leq (m - 1)/(\rho_I - 1)$. 
\end{lemma}
\begin{proof}
By our assumption that $OPT(I) = 1$, $\sum_i t(i) \leq \sum_p s(p)$. Also, as we have $\rho_I s(p) \leq w_I'(p) + 1$ from Lemma \ref{sec21}, $\rho_I \sum_p s(p) \leq \sum_{i \neq n} t(i) + m$. By the assumption that $t(n) = 1$,
$\rho_I \sum_i t(i) \leq \rho_I \sum_p s(p) \leq \sum_{i \neq n} t(i) + m = \sum_{i} t(i) + m - 1.$
By rearranging the inequality, we obtain the lemma statement. \qed
\end{proof}

We are now ready to prove our main theorem. 
\begin{theorem}
When $m = 3,4,5$, the worst-case approximation ratio of LPT on uniform processors is $\rho_m$.
\end{theorem}
\begin{proof}
Let $I = (3,n,s,t)$ be an instance with $\rho_I \geq \rho_3 > 1.38$. By Lemma \ref{lem15}, $I$ is not minimal when $n > (3 - 1)/(1.38 - 1) > 5.27$.
We then know that, when $n \geq 6$, there is another instance $I' = (m', n', s', t')$ such that $\rho_{I'} \geq \rho_I$, $m' \leq m$, $n' \leq n$, and $(m',n') \neq (m,n)$. Therefore, to consider the worst approximation ratio when $m = 3$, it is enough to consider instance $I' = (m', n', s', t')$ such that $m' \leq 3$ and $n' \leq 5$. 

It is clear that, for any $I'$ such that $m' = 1$, $\rho_{I'} = 1$. Furthermore, we know from Lemma \ref{lem1} that all $\rho_{I'}$ such that $m' = 2$ gives $\rho_{I'} \leq \rho_2 < \rho_3$. They do not give the worst approximation ratio.
It is then enough to consider instance $I'$ such that $m' = 3$. By Lemma \ref{sec31} and Theorem \ref{thm2}, we know that $I'$ do not give the worst ratio when $n' \leq 3$ or $n' = 5$. The worst ratio is then attained when $n' = 4$. We then can obtain by Theorem \ref{thm1} that the ratio is $\rho_m$.

We can use the same argument for $m = 4$. An instance $I = (4,n,s,t)$ is optimal only if $n \leq (4 - 1)/(\rho_4 - 1) < 6.98$. It is enough to consider instance $I' = (m', n', s', t')$ such that $m' \leq 4$ and $n' \leq 6$. By the previous paragraph, $\rho_{I'} \leq \rho_3 < \rho_4$ when $m' \leq 3$ and, by Lemma \ref{sec31} and Theorem \ref{thm2}, the worst ratio is not attained when $n' \leq 4$ or $n' = 6$. We then can conclude that the ratio is attained when $m' = 4$ and $n' = 5$. By Theorem \ref{thm1}, the ratio is $\rho_m$.

For $m = 5$, because $(5 - 1)/(\rho_5 - 1) < 8.89$, we can use the similar argument to show that the worst ratio is attained when $m = 5$ and $n \in \{6, 8\}$. It is then enough to show that any minimal instance $I = (5, 8, s, t)$ does not give the worst approximation ratio. Assume a contradictory statement that $\rho_I > \rho_m$. Because $8 = n \leq \sum_i t(i) < 8.89$ and $t(i) \geq 1$, for any set of task $T$, $\sum_{i \in T} t(i) < |T| + 0.89.$

Now, let us consider a processor $p$ such that $|T_I^*(p)| \geq 2$. If $|T'_I(p)| < |T^*_I(p)|$, then, by Lemma \ref{sec21}, 
\begin{eqnarray*}1.45 & < & \rho_5 < \rho_I \leq \frac{1 + w'_I(p)}{s(p)}=\frac{1 + \sum_{i \in T'_I(p)}t(i)}{s(p)} \leq \frac{1 + |T'_I(p)| + 0.89}{s(p)} \\ 
& \leq & \frac{1 + |T^*_I(p)|  - 1 + 0.89}{s(p)} \leq 1 + 0.89/2 = 1.445.
\end{eqnarray*}
Therefore, $|T'_I(p)| \geq |T^*_I(p)|$ for $p$ such that $|T_I^*(p)| \geq 2$. For processor $p$ such that $|T_I^*(p)| = 1$, by Lemma \ref{sec31}, we have $|T_I'(p)| \geq 1 \geq |T_I^*(p)|$. Hence, $|T_I'(p)| \geq |T_I^*(p)|$ for all $p$. However, this contradicts the fact that there is one less task in $\sum_p |T_I'(p)|$ than in $\sum_p |T_I^*(p)|$. \qed
\end{proof}

\section{Conclusion}
In this work, we show the tight approximation ratio of the LPT algorithm for $m$ uniform processors when $m = 3,4,5$. On the way to show that, we found several results which give us a deeper understanding of the algorithm. Those results include Lemma \ref{lem15} where we show that the worst approximation ratio is obtained when the number of tasks is small compared to the number of processors, and Theorems \ref{thm1} and \ref{thm2} where we give results for the case when the number of tasks is $m + 1$ and $m + 2$. We believe that those results will play an important role in future analyses of the algorithm. We also found that the analysis for cases with $m + 2$ tasks is much more complicated than that of $m + 1$ tasks. It is very clear that the analysis for $m + 3$ tasks, which would need for having the tight ratio for $m \geq 6$, would be much more complicated than both of the analyses. 
It would be very complicated to use this analysis method for having similar results for larger $m$.
\vspace{0.2cm}

\noindent\textbf{Acknowledgement:} The authors would like to thank Taku Onodera for several useful comments and ideas. 

\vspace{-0.2cm}
\bibliographystyle{splncs04}
\bibliography{ref}

\begin{thebibliography}{10}
\providecommand{\url}[1]{\texttt{#1}}
\providecommand{\urlprefix}{URL }
\providecommand{\doi}[1]{https://doi.org/#1}

\bibitem{chen1991parametric}
Chen, B.: Parametric bounds for {LPT} scheduling on uniform processors. Acta
  Mathematicae Applicatae Sinica  \textbf{7}(1),  67--73 (1991)

\bibitem{coffman1978application}
Coffman, Jr, E.G., Garey, M.R., Johnson, D.S.: An application of bin-packing to
  multiprocessor scheduling. SIAM Journal on Computing  \textbf{7}(1),  1--17
  (1978)

\bibitem{dobson1984scheduling}
Dobson, G.: Scheduling independent tasks on uniform processors. SIAM Journal on
  Computing  \textbf{13}(4),  705--716 (1984)

\bibitem{friesen1987tighter}
Friesen, D.K.: Tighter bounds for {LPT} scheduling on uniform processors. SIAM
  Journal on Computing  \textbf{16}(3),  554--560 (1987)

\bibitem{garey1978strong}
Garey, M.R., Johnson, D.S.: {``strong''}np-completeness results: Motivation,
  examples, and implications. Journal of the ACM (JACM)  \textbf{25}(3),
  499--508 (1978)

\bibitem{ghalami2019scheduling}
Ghalami, L., Grosu, D.: Scheduling parallel identical machines to minimize
  makespan: A parallel approximation algorithm. Journal of Parallel and
  Distributed Computing  \textbf{133},  221--231 (2019)

\bibitem{gonzalez1977bounds}
Gonzalez, T., Ibarra, O.H., Sahni, S.: Bounds for {LPT} schedules on uniform
  processors. SIAM journal on Computing  \textbf{6}(1),  155--166 (1977)

\bibitem{graham1966bounds}
Graham, R.L.: Bounds for certain multiprocessing anomalies. Bell system
  technical journal  \textbf{45}(9),  1563--1581 (1966)

\bibitem{graham1969bounds}
Graham, R.L.: Bounds on multiprocessing timing anomalies. SIAM journal on
  Applied Mathematics  \textbf{17}(2),  416--429 (1969)

\bibitem{hochbaum1986polynomial}
Hochbaum, D.S., Shmoys, D.B.: A polynomial approximation scheme for machine
  scheduling on uniform processors: using the dual approximation approach. In:
  FSTTCS'86. pp. 382--393. Springer (1986)

\bibitem{jansen2020approximation}
Jansen, K., Lassota, A., Maack, M.: Approximation algorithms for scheduling
  with class constraints. In: SPAA'20. pp. 349--357 (2020)

\bibitem{kovacs2009tighter}
Kov{\'a}cs, A.: Tighter approximation bounds for {LPT} scheduling in two
  special cases. Journal of Discrete Algorithms  \textbf{7}(3),  327--340
  (2009)

\bibitem{kovacs2010new}
Kov{\'a}cs, A.: New approximation bounds for {LPT} scheduling. Algorithmica
  \textbf{57}(2),  413--433 (2010)

\bibitem{sahni1976algorithms}
Sahni, S.K.: Algorithms for scheduling independent tasks. Journal of the ACM
  (JACM)  \textbf{23}(1),  116--127 (1976)

\bibitem{suda2006}
Suda, R.: A survey on task scheduling for heterogeneous parallel computing
  environments (in japanese). IPSJ Trans. on Advanced Computing Systems
  \textbf{47}(SIG 18),  92--114 (2006)

\end{thebibliography}

\end{document}